\documentclass[11pt]{article}

\usepackage{mathrsfs}           
\usepackage{algorithm}
\usepackage{algorithmic}

\usepackage{enumerate}

\usepackage{amsmath,amssymb}    
\usepackage{verbatim}           
\usepackage{xspace}             
\usepackage{graphicx,float}     
\usepackage{ifthen,calc}        
\usepackage{textcomp}           
\usepackage{fancybox}           
\usepackage{hhline}             
\usepackage{float}              

\usepackage[rflt]{floatflt}     
\usepackage{subfigure}

\usepackage{color}
\definecolor{ForestGreen}{rgb}{0.1333,0.5451,0.1333}
\usepackage{thumbpdf}

\floatstyle{ruled}
\newfloat{algo}{tbp}{lop}
\floatname{algo}{Algorithm}

\usepackage{float}
\floatstyle{plain}\newfloat{myfig}{t}{figs}[section]
\floatname{myfig}{\textsc{Figure}}
\floatstyle{plain}\newfloat{myalg}{H}{algs}[section]
\floatname{myalg}{}
\usepackage{fullpage}
\usepackage[utf8]{inputenc}
\usepackage[english]{babel}
\usepackage{amsmath,amsthm,amssymb,stackrel}
 \usepackage{todonotes}
\usepackage{url}
 \usepackage{relsize}
 \usepackage{authblk}

\usepackage[authoryear]{natbib}
\setcitestyle{square,aysep={},yysep={;}}
\newcommand{\citeN}[1]{\citet{#1}}
\renewcommand{\cite}{\citep}

\newtheorem{claim}{Claim}
\newtheorem{lemma}{Lemma}
\newtheorem{theorem}{Theorem}

\newtheorem{proposition}{Proposition}

 \newcommand{\SI}{\textsc{Subgraph Isomorphism}\xspace}
  \newcommand{\LIGH}{\textsc{Locally Injective Graph Homomorphism}\xspace}
  \newcommand{\GM}{\textsc{Graph Minor}\xspace}
    \newcommand{\QAP}{\textsc{Quadratic Assignment Problem}\xspace}    
    \newcommand{\TGM}{\textsc{Topological Graph Minor}\xspace}  
   \newcommand{\MDE}{\textsc{Minimum Distortion Embedding}\xspace}
\newcommand{\LGH}{\textsc{List Graph Homomorphism}} 
\newcommand{\cO}{\mathcal{O}}
\newcommand{\Oh}{\cO}
\newcommand{\cOs}{\mathcal{O}^*}
\newcommand{\cL}{\mathcal{L}}

\renewcommand{\deg}{\operatorname{deg}}

\newcommand{\GH}[0]{\textsc{Graph Homomorphism}\xspace}

\def\cqedsymbol{\ifmmode$\lrcorner$\else{\unskip\nobreak\hfil
\penalty50\hskip1em\null\nobreak\hfil$\lrcorner$
\parfillskip=0pt\finalhyphendemerits=0\endgraf}\fi}

\newcommand{\defproblemu}[3]{
  \vspace{2mm}
  \vspace{1mm}
\noindent\fbox{
  \begin{minipage}{0.95\textwidth}
  #1 \\
  {\bf{Input:}} #2  \\
  {\bf{Task:}} #3
  \end{minipage}
  }
  \vspace{2mm}
}

\newcommand{\gr}[1]{\tilde{#1}} 

\title{Tight Lower Bounds on Graph Embedding Problems\thanks{The research leading to these results has received funding from the Government of the Russian Federation (grant 14.Z50.31.0030). The research of Alexander Golovnev is supported by NSF grant 1319051. The research of Alexander Kulikov is also supported by  the grant of the President of Russian Federation  (MK-6550.2015.1). The research of Jakub Pachocki is supported by NSF grant CCF-1065106.
The research of Marek Cygan and Arkadiusz Soca\l{}a is supported by National Science Centre of Poland, Grant Number UMO-2013/09/B/ST6/03136. Preliminary versions of this work were presented at ICALP 2015 and SODA 2016.
}}
\author[1]{Marek~Cygan}
\author[2,5]{Fedor~V.~Fomin}
\author[3,5]{Alexander~Golovnev}
  \author[5]{Alexander~S.~Kulikov}
\author[4]{Ivan~Mihajlin}
\author[6]{Jakub~Pachocki}
\author[1]{Arkadiusz~Soca\l{}a}
 \affil[1]{Institute of Informatics, University of Warsaw, Poland}
  \affil[2]{University of Bergen, Norway}
    \affil[3]{New York University, USA}
      \affil[4]{University of California---San Diego, USA}
        \affil[5]{St.~Petersburg Department of Steklov Institute of Mathematics of the Russian Academy of Sciences, Russia}
           \affil[6]{Carnegie Mellon University, USA}

\date{}

\begin{document}

\maketitle

\begin{abstract}
We prove that unless the Exponential Time Hypothesis (ETH) fails, deciding if there is a homomorphism from graph $G$ to  graph $H$ cannot be done in time $|V(H)|^{o(|V(G)|)}$. We also show an exponential-time reduction from Graph Homomorphism to Subgraph Isomorphism. This rules out (subject to ETH) a possibility of $|V(H)|^{o(|V(H)|)}$-time algorithm deciding if graph $G$ is a subgraph of~$H$. 
For both problems our lower bounds asymptotically match the running time of brute-force algorithms trying all possible mappings of one graph into another. Thus, our work  closes  
 the   gap in the known complexity of these fundamental problems.

Moreover, as a consequence of our reductions conditional lower bounds follow for other related problems such as Locally Injective Homomorphism, 
Graph Minors, Topological Graph Minors, Minimum Distortion Embedding and Quadratic Assignment Problem.
 \end{abstract}

\section{Introduction}\label{sec:intro}


We establish tight conditional lower bounds on the complexity of several fundamental graph embedding problems including
\GH, \SI, \GM, \TGM, and \MDE. 
For   given undirected graphs $G$ and $H$, all these problems can be solved in time $n^{\cO(n)}$ by a brute-force algorithm that tries all possible embeddings of $G$ into $H$, where $n$ is the number of vertices in $G$ and $H$.  We show that unless 
 the Exponential Time Hypothesis (ETH) fails, the running time  $n^{\cO(n)}$ is unavoidable. This resolves a number of open problems about graph embeddings that can be found in the literature. We start by  defining embedding problems and providing for each of the problems a brief overview of the related previous results.

\paragraph{\GH}
A {\em homomorphism} $G\to H$ from an undirected graph $G$ to an undirected graph $H$ is a mapping
from the vertex set   $G$ to that of $H$ such that the image of every edge of $G$ is an edge of $H$.  
In other words, there is $G\to H$ if and only if there exists a mapping $g : V(G) \to V(H)$,
such that for every edge $uv \in E(G)$,  we have $g(u)g(v) \in E(H)$.  Then the \GH problem HOM$(G,H)$
 is defined as follows. 
 
\defproblemu{\GH}{Undirected graphs $G$ and $H$. } {Decide whether there is a homomorphism  $G\to H$.}

Many combinatorial structures in $G$, for example cliques, independent sets,  and proper vertex colorings,
may be viewed as graph homomorphisms to a particular graph $H$, see the  book of 
~\citeN{HellN04}
for a thorough introduction to the topic. 
 It is well-known that  \textsc{Coloring} is a special case of graph homomorphism. More precisely, a graph $G$ can be colored with at most $h$ colors if and only if $G\to K_h$, where $K_h$ is a complete graph on $h$ vertices. Due to this, very often in the literature HOM($G,H$), when $h=|V(H)|$,  is referred as $H$-coloring of~$G$.
It was shown by 
\citeN{FederV98}  that   \textsc{Constraint Satisfaction Problem}  (CSP) can be interpreted  as a homomorphism problem on relational structures, and thus  \textsc{Graph Homomorphism}  encompasses a large family of problems generalizing  \textsc{Coloring}  but less general  than CSP.

Hell and Ne\v set\v ril showed that for every fixed simple graph $H$,
the problem whether there exists a homomorphism from $G$ to $H$  is solvable in
polynomial time if $H$ is bipartite,
and NP-complete if $H$ is not bipartite~\cite{HellN90-On}. 
Since then, algorithms for and the complexity of graph homomorphisms (and homomorphisms between other discrete structures) have been   studied intensively
\cite{Austrin10,BartoKN08,Grohe07,Marx10,Raghavendra08}.

There are two different ways  graph homomorphisms are used to extract useful information about graphs. Let us consider two homomorphisms, from a ``small" graph $F$ into a ``large'' graph $G$ and from a ``large'' graph $G$ into a  ``small" graph $H$, which can be represented by the following formula  (here we borrow  the intuitive description from the  book of 
~\citeN{lovasz2012large})
\[
F\to {\mathlarger{\mathlarger G}} \to H.
\]
Then ``left-homomorphisms" from various small graphs $F$ into $G$ are useful to study the local structure of $G$. For example, if $F$ is a triangle, then the number of  ``left-homomorphisms" from $F$ into $G$ is the number of triangles in graph $G$. 
This type of information is closely related  to sampling, and we refer to the book of 
\citeN{lovasz2012large} which provides many applications of homomorphisms.  ``Right-homomorphisms" into ``small" different graphs $H$  are related to global properties of graph~$G$. 

The trivial brute-force algorithm solving   ``left-homomorphism" from an $f$-vertex graph $F$ into an $n$-vertex graph $G$ runs in time $2^{\cO(f \log{n})}$:
 we try all possible vertex subsets of $G$ of size at most $f$, which is $n^{\cO(f)}$ and then for each subset try all possible $f^f$ mappings into it from $F$.  Interestingly,  this na\"{\i}ve  algorithm is asymptotically optimal.  Indeed, 
as it was shown by 
\citeN{Chen20061346}, assuming Exponential Time Hypothesis, there is no $g(k)n^{o(k)}$ time algorithm deciding if an input $n$-vertex graph $G$ contains a clique of size at least $k$,  for any computable function $g$.  Since this is a very special case of  \GH     with  $F$ being a clique of size $k$, the result of  Chen et al. rules out algorithms for  \GH of running time $g(f)2^{o(f\log{n})}$, from $F$ to $G$, when the number of vertices $f$ in $F$ is significantly smaller than the number of vertices $n$ in $G$. 
 
%
%
%
%
 
 The  interest in  ``right-homomorphisms"   is due to the recent developments in the area of exact exponential algorithms for \textsc{Coloring} and \textsc{$2$-CSP} (\textsc{CSP} where all constraints have arity at most~$2$) problems. 
The area of exact exponential algorithms is about solving intractable problems significantly faster than the trivial exhaustive search, though still in exponential time \cite{FominKratschbook10}. For example, as for  \GH, a na\"{\i}ve  brute-force algorithm for coloring an $n$-vertex  graph $G$ in $h$ colors is to try for every vertex a possible color, resulting in the running time $\cOs(h^n)=2^{\cO(n\log{h})}$.\footnote{$\cOs(\cdot)$ hides polynomial factors in the input length. Most of the algorithms considered in this paper take graphs $G$ and $H$ as an input. By saying that such an algorithm has a running time $\cOs(f(G,H))$, we mean that the running time is upper bounded by
$(|V(G)|+|E(G)|+|V(H)|+|E(H)|)^{\cO(1)}\cdot f(G,H)$.} 
Since $h$ can be of order $\Omega(n)$, the brute-force algorithm computing the chromatic number  runs in time $2^{\cO(n\log{n})}$. It was already observed in 1970s by 
\citeN{Lawler76} that the brute-force for  the  \textsc{Coloring}  problem can be  beaten by making use of dynamic programming over maximal independent sets  resulting in single-exponential  running time  $\cOs((1+\sqrt[3]{3})^n)=\cO( 2.45^n)$.
  Almost 30 years later 
%
 \citeN{BjorklundHK2009-Se} succeeded to reduce the running time to $\cOs( 2^n)$.  And as we observed already,  for $H$-coloring, the brute-force algorithm solving 
 $H$-coloring runs in time $2^{\cO(n\log{h})}$. In spite of  all the similarities between graph coloring and homomorphism, no substantially faster algorithm was known and it was an open question in the area of exact algorithms if there is a single-exponential algorithm solving $H$-coloring in time  $2^{\cO(n+h)}$ \cite{FHK2007,Rzazewski14,Wahlst10,W2011}, see also \cite[Chapter 12]{FominKratschbook10}.

On the other hand, \GH is a special case of \textsc{2-CSP} with $n$  variables and domain of size $h$.  It was shown by 
\citeN{T2008} that unless the Exponential Time Hypothesis fails, there is no algorithm solving 
\textsc{2-CSP} with $n$ variables and domain of size $h$ in time  $h^{o(n)}=2^{o(n\log h)}$. This excludes (up to ETH)  the existence of a single-exponential $c^n$ time algorithm for some constant $c>1$ for \textsc{2-CSP}.

Another interesting variant of \GH is related to  graph labelings. 
A homomorphism $f \colon  G \to H$ is called \emph{locally injective} if for every vertex 
$u \in V (G)$, its neighborhood is mapped injectively into the neighborhood of $f (u)$ in $H$, i.e., if every two vertices with a common neighbor in $G$ are mapped onto distinct vertices in $H$. 
 
\defproblemu{\LIGH}{Undirected graphs $G$ and $H$. } {Decide whether there is a locally injective homomorphism  $G\to H$.}

As graph homomorphism generalizes graph coloring, locally injective graph homomorphism can be seen as a generalization of graph 
 distance constrained labelings. An $L(2, 1)$-labeling of a graph $G$ is a mapping from $V(G)$ into the nonnegative integers such that the 
labels assigned to vertices at distance $2$ are different 
 while labels 
 assigned to adjacent vertices differ by at least $2$. This problem was studied intensively in combinatorics and algorithms, see, e.g.,
\citeN{Griggs:1992uq} and  
 \citeN{FialaGK08}. 
 Fiala and  Kratochv\'{\i}l suggested the following generalization of $L(2, 1)$-labeling, we refer  \cite{FialaK08} for the survey. For graphs $G$ and $H$, an $H(2,1)$-labeling is a mapping $f : V(G)\to V(H)$ such that for every pair of distinct adjacent vertices $u,v\in V(G)$, images  $f(u)$ $f(v)$ are distinct and nonadjacent in $H$. Moreover, if the distance between $u$ and $v$ in $G$ is two, then $f(u)\neq f(v)$. It is easy to see that a graph $G$ has an $L(2,1)$-labeling  with maximum label at most $k$ if and only if there is an $H(2,1)$-labeling for $H$ being a $k$-vertex path. Then the following is known, see for example \cite{FialaK08}: there is an $H(2,1)$-labeling of a graph $G$ if and only if there is a locally injective homomorphism from  $G$ to the complement of $H$.
 
 Several single-exponential algorithms for  $L(2,1)$-labeling can be found in the literature, the most recent algorithm is due to  
\citeN{Junosza-SzaniawskiKLRR13} which runs in time $\cO(2.6488^n)$. For $H(2,1)$-labeling, or equivalently for locally injective homomorphisms,   single-exponential algorithms were known only for special cases when the maximum degree of $H$ is bounded  \cite{HavetKKKL11} or  when the bandwidth of the complement of $H$ is bounded \cite{Rzazewski14}. 

\paragraph{\SI}
We say that an undirected $G$ is a \emph{subgraph} of $H$ if one can remove some edges and vertices
of $H$, so that what remains is isomorphic to $G$. In other words, $G$ is a subgraph of $H$ if and only if there exists an injective mapping $g : V(G) \to V(H)$,
such that for each edge $uv \in E(G)$,  $g(u)g(v) \in E(H)$.
We define

\defproblemu{\SI}{Undirected graphs $G$ and $H$. } {Decide whether $G$ is a subgraph of $H$.}

\SI is an important and very general problem. 
  Several flagship graph problems can be viewed as instances of \SI:
\begin{itemize}
  \item {\sc Hamiltonicity}$(G)$: Is $C_n$ (a cycle with $n$ vertices) a subgraph of $G$?
  \item {\sc Clique}$(G,k)$: Is $K_k$ a subgraph of $G$?
  \item {\sc 3-Coloring}$(G)$: Is $G$ a subgraph of $K_{n,n,n}$, a tripartite graph with $n$ vertices in each of its three independent sets?
\item   {\sc Bandwidth}$(G,k)$: Is $G$ a subgraph of $P_n^k$ (a $k$-th power of an $n$-vertex path)?
\end{itemize}
 
All of the mentioned problems are NP-complete, and the best known algorithms
for all the listed special cases work in exponential time.
In fact, all those problems are well-studied from the exact exponential algorithms perspective~\cite{BeiEpp05,hamiltonicity,mis,bw,Feige00,HeldKarp62,Lawler76,Robson86,TarjanTrojanowski77},
where the goal is to obtain an algorithm of running time $\Oh(c^n)$ for the smallest possible value of $c$.
Furthermore, the \SI problem was very extensively studied from the viewpoint of fixed parameter tractability, see~\cite{michal} for a discussion of 19 different possible parametrizations.
All the mentioned special cases of \SI admit $\Oh(c^n)$ time algorithms, by using either branching,
inclusion-exclusion principle, or dynamic programming.
On the other hand, a simple exhaustive search for the \SI problem---numerating all possible mappings from the pattern graph to the host graph---runs in $2^{\Oh(n \log n)}$ time, 
where $n$ is the total number of vertices of the host graph and pattern graph.

Therefore, a natural question is whether \SI admits an $\Oh(c^n)$ time algorithm.
This was repeatedly posed as an open problem~\cite{bedlewo,counting-homo,dagstuhl-1,dagstuhl-2}. 
In particular, in the monograph of 
~\citeN{FominKratschbook10}   the existence of $\Oh(c^n)$ time algorithm
for \SI  was put among the few questions in the open problems section.

\SI is  a special case of \QAP,  which is

\defproblemu{\QAP (QAP)}{$n\times n$ matrices $A=(a_{ij})$ and $B=(b_{ij})$ with real entries. } {Find a permutation $\pi$ minimizing  $\sum_{i=1}^n \sum_{j=1}^n a_{\pi(i)\pi(j)}b_{ij}$.}
%
 
 Indeed, $G$ is a subgraph of $H$ if and only if for the instance of  QAP with 
$A$ and $B$ being adjacency matrices of $G$ and the complement of $H$ the optimum value is $0$\footnote{If $G$ has smaller number of vertices than $H$, then it should be first padded with isolated vertices to make the number of vertices in both graphs equal.}.
Problem 7.6 in the  influential survey of Woeginger  on exact algorithms \cite{Woeginger03} is to prove that QAP cannot be solved in time $\cO(c^n)$ for any fixed value $c$ (under some reasonable assumption).

\paragraph{\GM}
For a graph $G$ and an edge $uv \in G$,
we define the operation of \emph{contracting  edge $uv$} as follows:   we delete vertices $u$ and $v$ from $G$, and add  a new vertex $w_{uv}$ adjacent to   all vertices that $u$ or $v$ was adjacent to in $G$.
We say that a graph $G$ is a \emph{minor} of $H$,  if $G$ can be obtained from some subgraph of $H$ by a series of edge contractions.
Equivalently, we may say that $G$ is a minor of $H$ if $G$ can be obtained from $H$ itself
by a series of edge deletions, edge contractions and vertex deletions.

\defproblemu{\GM}{Undirected graphs $G$ and $H$. } {Decide whether $G$ is a minor of $H$.}

\GM is a fundamental  problem in graph theory and graph algorithms.  
 By the  theorem of 
\citeN{RobertsonS-GMXIII},  there exists a computable function $f$ and an algorithm that, for  given graphs $G$ and $H$,
checks in time  $f(G) |V(H)|^3$   whether $G$ is a minor of $H$. 
However, when the size of the graph $G$ is not a constant, nothing beyond a brute-force algorithm trying all possible partitions of a vertex set of $H$ was known.  

Related notion of graph embedding is the notion of topological minor. We say that a graph $G$ is a subdivision of a graph $H$ if $G$ can be obtained from $H$ by contracting only edges incident with vertices of degree two. In other words, $G$ is obtained from $H$ by replacing edges with paths. 
A~graph $G$ is called a \emph{topological minor} of a graph $H$  if a subdivision of $G$ is isomorphic to a subgraph of $H$.

\defproblemu{\TGM}{Undirected graphs $G$ and $H$. } {Decide whether $G$ is a topological minor of $H$.}

\citeN{LingasW09}  gave an algorithm of running time  $\cOs({{n}\choose{p}} p! 2^{n-p})$ solving \TGM for 
  $n$-vertex graph $H$ and $p$-vertex graph $G$.

\paragraph{\MDE}
Given an undirected connected graph $G$ with the vertex set $V(G)$ and the edge set $E(G)$, the {\em graph metric} of $G$ is $M(G) = (V(G),D_G)$, where the
distance function $D_G$ is the shortest path distance between $u$ and $v$ for every pair of vertices $u,v \in V(G)$. Given a graph metric $M$
and another metric space $M'$ with distance functions $D$ and $D'$, a mapping $f:M \rightarrow M'$ is called an {\em embedding} of $M$ into $M'$. The mapping $f$  is \emph{non-contracting}, if 
  for every pair of points $p,q$ in $M$,
$D(p,q) \leq D'(f(p),f(q))$. The \emph{distortion} of embedding $f$ is the  minimum number  $d_f$ such that 
$D(p,q) \cdot d_f \geq D'(f(p),f(q))$.  We define 

\defproblemu{\MDE}{Undirected graphs $G$ and $H$. } {Find a non-contracting embedding of 
 $G$ into $H$ of minimum distortion.}
 
 Most of exact algorithms for \MDE deal with a special case when the host graph $H$ is a path or a tree of bounded degree  \cite{BadoiuCIS05,BadoiuDGRRRS05,bw,FellowsFLLRS13,FominLS11,KenyonRS09}. 
 In particular, an optimal-distortion embedding into a line can be found in time   $2^{\cO(n)}$ \cite{bw,FominLS11}.

 \paragraph{Our results.}
In this paper we show that from the algorithmic perspective,   the behavior of  ``right-homomorphism"  is, unfortunately, much closer to  \textsc{2-CSP} than to 
\textsc{Coloring}.  This result will also imply similar lower bounds for many other graph embedding and containment problems. 
All  lower bounds obtained in this paper are conditional, they hold unless 
 the Exponential Time Hypothesis
\cite{eth1,eth2} fails. ETH
is an established assumption; many interesting lower bounds have been found under this hypothesis  (see~\cite{CFKLMPPS2014,eth-survey} for surveys).
We formulate  ETH in the next section.

The first main result of this paper is the following theorem, which 
excludes  (up to ETH) resolvability of  HOM$(G,H)$ in time
$2^{o({n\log h})}$, thus resolving the open qestion from  \cite{FHK2007,Rzazewski14,Wahlst10,W2011}.

\newsavebox{\boxmainboundone}
\sbox\boxmainboundone{\parbox{\textwidth}{
\begin{theorem}\label{main:theorem_homs} 
Unless ETH fails, for any constant $d>0$ there exists a constant $c=c(d)>0$ such that
for any non-decreasing
function
$3\le h(n)\le n^d$,
there is no algorithm solving 
\GH
 from an $n$-vertex graph $G$ to a graph $H$ with at most $h(n)$ vertices in time 
\begin{equation}\label{eq:vert}
\cO(2^{cn\log{h(n)}}) \,.
\end{equation}
\end{theorem}
}}
\noindent\usebox{\boxmainboundone}

Let us remark that 
in order to obtain more general results, in all lower bounds proven in this paper we assume implicitly that the number $h$ of vertices of the graph $H$ is a function of the number $n$ of the vertices of the graph~$G$. At the same time, to exclude some pathological cases we assume that the function $h(n)$ is ``reasonable'' meaning that it is non-decreasing and time-constructible.

With a tiny modification,  the proof of Theorem~\ref{main:theorem_homs} can be adapted to show  a similar lower  bound  for \LIGH.

\newsavebox{\boxlocalbound}
\sbox\boxlocalbound{\parbox{\textwidth}{
\begin{theorem}\label{main:theorem_local_homs}
 Unless ETH fails, for any constant $d>0$ there exists a constant $c=c(d)>0$ such that
for any non-decreasing
function
$3\le h(n)\le n^d$,
there is no algorithm deciding if there is a locally injective homomorphism from an $n$-vertex graph $G$ to a graph $H$ 
with at most $h(n)$ vertices in time 
\( \cO(2^{cn\log{h(n)}}) \,. \)
\end{theorem}
 }}
\noindent\usebox{\boxlocalbound}

The second main result of this paper is about \SI, resolving the open question asked in ~\cite{bedlewo,counting-homo,dagstuhl-1,dagstuhl-2,FominKratschbook10}.

\begin{theorem}\label{main:theorem_SI} 
Unless ETH fails,
there is no algorithm   solving \SI for graphs $G$ and $H$ in time $2^{o(n \log n)}$, where $n=|V(G)|+|V(H)|$.
\end{theorem}

Theorem~\ref{main:theorem_SI} implies that QAP cannot be solved in time   $2^{o(n \log n)}$ unless ETH fails and hence provides the answer to the open problem of 
\citeN{Woeginger03}.

An important feature of our proof is that it rules out solvability of \SI in time $2^{o(n \log n)}$ even for the special case  when  $|V(G)|=|V(H)|=n$. 
Since in this special case a graph $G$ is a (topological) minor of $H$ if and only if $G$ is a subgraph of $H$. Thus the case of \GM and \TGM  when  $|V(G)|=|V(H)|=n$ cannot be resolved in time  $2^{o(n \log n)}$ as well. Similar arguments work for various modifications of \GM like \textsc{Shallow Graph Minor}, etc. 

To see how the bound on \SI yields the bound on \MDE, we observe that 
an $n$-vertex graph $G$ admits a non-contracting embedding of distortion $1$ into an $n$-vertex graph $H$ if and only if $H$ is a subgraph of $G$.

\paragraph{Methods}
  To establish lower bounds for graph homomorhisms,   we proceed in two steps. 
 First we obtain lower bounds for   \textsc{List Graph Homomorphism} by reducing it to the $3$-coloring problem on graphs of bounded degree. More precisely, for a given graph $G$ with vertices of small degrees, we construct an instance $(G',H')$ 
 of  \textsc{List Graph Homomorphism}, such that $G$ is $3$-colorable if and only if there exists a list homomorphism from $G'$ to $H'$. Moreover, our construction guarantees that a ``fast" algorithm for list homomorphism 
implies an algorithm for $3$-coloring violating ETH.
 The reduction is based on a  ``grouping" technique, however, to do the required grouping we need a trick exploiting the condition that $G$ has a bounded maximum vertex degree and thus can be colored in a bounded number of colors in polynomial time. In the second step of reductions we proceed from list homomorphisms to normal homomorphisms. Here we need specific  gadgets  with a property that 
  any homomorphism from such a graph to itself preserves an order
  of its specific structures.

  The remaining part of the paper is organized as follows. Section~\ref{sec:prelim} contains all necessary definitions. In Section~\ref{sec:lemmata} we give technical lemmata and reductions which are used to prove lower bounds for the \GH{} problem in Section~\ref{sec:hom} and for the \SI{} in Section~\ref{sec:subiso}. We conclude with some open problems in Section~\ref{sec:conop}.

%
%

%

\section{Preliminaries}\label{sec:prelim}

\paragraph{Graphs}
We consider simple undirected graphs, where $V(G)$ denotes the set
of vertices and $E(G)$ denotes the set of edges of a graph $G$.
For a given subset $S$ of $V(G)$, $G[S]$ denotes the subgraph of $G$ induced by $S$,
and $G-S$ denotes the graph $G[V(G)\setminus S]$.
A vertex set $S$ of $G$ is an {\em independent set} if $G[S]$ is a graph
with no edges, and
$S$ is a {\em clique} if $G[S]$ is a complete graph.
The set of neighbors of a vertex $v$ in $G$ is denoted by $N_G(v)$, and the
set of neighbors of a vertex set $S$ is $N_G(S) = \bigcup_{v \in S}N_G(v)
\setminus S$. By $N_G[S]$ we denote the closed neighborhood of the set $S$, i.e., the set $S$ together with all its neighbors: $N_G[S]=S \cup N_G(S)$. For an integer $n$, we use $[n]$ to denote the set of integers $\{1,\dots, n\}$. 

The complete graph on $k$ vertices is denoted by $K_k$.
A {\em coloring} of a graph $G$ is a function assigning a color to each
vertex of $G$ such that adjacent vertices have different colors.
A $k$-coloring of a graph uses at most $k$ colors, and the \emph{chromatic number}   $\chi (G)$ is the smallest number
of colors in a coloring of $G$. 
By Brook's theorem, 
for  any connected  graph $G$ with maximum degree $\Delta>2$,  the chromatic number of $G$ is at most $\Delta$ unless $G$ is a complete graph, in which case the chromatic number is $\Delta + 1$. Moreover, 
a $(\Delta +1)$-coloring of a  graph can be found in polynomial time by a straightforward
 greedy algorithm. 

Throughout the paper we implicitly assume that there is a total order on the set of vertices of a given graph. This allows us to treat a $k$-coloring of a $n$-vertex graph simply as a vector in~$[k]^n$.

%

Let $G$ be an $n$-vertex graph, $1 \le r \le n$ be an integer, and $V(G)=B_1 \sqcup B_2 \sqcup \ldots \sqcup B_{\lceil \frac nr \rceil}$ be a partition of the set of vertices of $G$. Then {\em the grouping} of $G$ with respect to the partition $V(G)=B_1 \sqcup B_2 \sqcup \ldots \sqcup B_{\lceil \frac nr \rceil}$ is a graph $G_r$
with vertices $B_1, \ldots, B_{\lceil \frac nr \rceil}$ such that $B_i$ and $B_j$
are adjacent if and only if there exist $u \in B_i$ and $v \in B_j$ such that $uv \in E(G)$. To distinguish vertices of the graphs $G$ and $G_r$, the vertices of $G_r$
will be called {\em buckets}.  

For a graph $G$, its {\em square $G^2$} has the same set of vertices as $G$
and $uv \in E(G^2)$ if and only if there is a path of length at most $2$ between $u$
and $v$ in $G$ (thus, $E(G) \subseteq E(G^2)$). It is easy to see that if the degree of $G$ is less than $\Delta$ then the degree of $G^2$ is less than $\Delta^2$ and hence a $\Delta^2$-coloring of $G^2$ can be easily found.

 \paragraph{Homomorphisms and list homomorphisms}
 Let $G$ and $H$ be  graphs.  A mapping 
 $\varphi : V(G)\to V(H)$ is a \emph{homomorphism} if for every edge $uv \in E(G)$ its image $\varphi(u)\varphi(v)\in E(H)$.
 If there exists a homomorphism from $G$ to $H$, we often write $G\to H$.
  The \textsc{ Graph Homomorphism} problem  HOM$(G,H)$
  asks whether or not  $G\to H$.

  Assume that for each vertex $v$ of   $G$   we are given a list $\cL(v) \subseteq V (H)$. A \emph{list homomorphism} of $G$ to $H$, also known as  a list $H$-coloring of $G$, with respect to the lists $\cL$, is a homomorphism  $\varphi : V(G)\to V(H)$, such that $\varphi(v) \in \cL(v)$ for all $v\in V (G)$. 
The \textsc{List Graph Homomorphism} problem LIST-HOM$(G,H)$
  asks whether or not  graph $G$ with lists $\cL$ admits a list homomorphism to $H $ with respect to $\cL$.
 
\paragraph{Exponential Time Hypothesis} 
Our lower bounds are based on a well-known complexity hypothesis formulated by  
\citeN{ImpagliazzoPZ01}.
 
\begin{quote}
\textbf{Exponential Time Hypothesis (ETH)}:  There is a constant $q>0$ such that 3-CNF-SAT with $n$ variables and $m$ clauses cannot be solved in time $2^{qn}(n+m)^{\cO(1)}$.
\end{quote}

This hypothesis is widely applied in the theory of exact exponential algorithms, we refer to \cite{CFKLMPPS2014,LMS2013} for an overview of ETH and its implications. 
 
In this paper we  use the following well-known application of ETH with respect to \textsc{$3$-Coloring} (see, e.g., Theorem~$3.2$ in~\cite{LMS2013}, and Exercise $7.27$ in~\cite{S2005}).
The \textsc{$3$-Coloring} problem is the problem to decide whether the given graph can be properly colored in $3$ colors.

\begin{proposition}
\label{prop:col}
Unless ETH fails, there exists a constant $q>0$ such that  \textsc{$3$-Coloring} on $n$-vertex graphs of average degree four cannot be solved in time $\cOs\left(2^{q n} \right)$.
\end{proposition}

It is well known that  \textsc{$3$-Coloring} remains NP-complete on graphs of maximum vertex degree four. Moreover, the classical reduction, see e.g. \cite{GareyJ79}, allows for a given $n$-vertex graph $G$ to construct a graph $G'$ with maximum vertex degree at most four  and   $|V(G')|=\cO(|E(G)|)$ such that $G$ is $3$-colorable if and only if $G'$ is.  Thus Proposition~\ref{prop:col} implies the following (folklore) lemma which will be used in our proofs. 
 
\begin{lemma}\label{lemma:3col}
\label{thm:3col}
Unless ETH fails, there exists a constant $q >0$  such that there is no algorithm solving \textsc{$3$-Coloring} on $n$-vertex  graphs of maximum degree four in time $\cOs\left(2^{q n} \right)$.
\end{lemma}

%
%
%
\section{Auxiliary Lemmata}\label{sec:lemmata}
In this section we provide reductions and auxiliary lemmata about colorings which will be used to prove lower bounds for \GH{} and  \SI.
\subsection{Balanced Colorings}\label{sec:colorings}
In the following we show how to construct a specific ``balanced" coloring of a graph
in polynomial time. Let $G$ be a graph of constant maximum degree. The coloring of $G$ we want to construct should satisfy three properties. First, it should be a proper coloring of~$G^2$. Then the size of each color class should be bounded as well as the number of edges between vertices from different color classes. More precisely, we prove the following lemma.

\begin{lemma}\label{lemma:coloring}
For any constant $d$, there exist constants $\alpha, \beta, \tau>1$ and a polynomial time algorithm
that for a given graph $G$ on $n$ vertices of maximum degree $d$ and an 
integer $\tau \le L \le \frac{n(d^2-1)}{2d^2(d^2+1)}$, finds a coloring $c \colon V(G) \to [L]$ satisfying the following properties:
\begin{enumerate}
\item The coloring $c$ is a proper coloring of~$G^2$.
\item There are only a few vertices of each color: 
for all  $i \in [L]$,
\begin{equation}\label{eq:balver}
|c^{-1}(i)| \le \left\lceil\alpha \cdot \frac{n}{L}\right\rceil \, .
\end{equation}
\item There are only a few  edges of $G$  between each pair of colors:
For all $i \neq j \in [L]$, we have  
\begin{multline*}\label{eq:fg}
k_{i,j} := |\{uv \in E (G) \colon c(u)=i, c(v)=j\}| \le 
K_{i,j} := \left\lceil\beta \cdot \frac{\min\{|c^{-1}(i)|, |c^{-1}(j)|\}}{L}\right\rceil \, .
\end{multline*}
\end{enumerate}
\end{lemma}
\begin{proof}
The algorithm starts by constructing greedily an independent set $I$ of $G^2$ of size $\left\lceil\frac{n}{d^2+1}\right\rceil$. 
Since the maximum vertex degree of $G^2$ does not exceed~$d^2$, this is always possible. We construct a partial coloring of $G^2$   by  coloring   the vertices of $I$ in $L$ colors  such that the obtained coloring is a balanced coloring of $G^2[I]$, meaning that the number of vertices of each color is $\lfloor |I|/L \rfloor$ or $\lceil |I|/L \rceil$.
Since $I$ is an independent set in $G^2$, such a coloring can be easily constructed in polynomial time.
In the obtained partial equitable coloring, we have that for every $i\in [L]$  
\begin{equation}\label{eq:gh}
|c^{-1}(i)| \ge \left\lfloor\frac{n}{L(d^2+1)}\right\rfloor \ge \frac{n}{2Ld^2} 
\end{equation}
(recall that $L \le \frac{n(d^2-1)}{2d^2(d^2+1)}$).
Let us note that the obtained precoloring of $G^2$ clearly satisfies the   first and the third conditions of the lemma. Since 
  the size of every $c^{-1}(i)$, $i\in [L]$, does not exceed $|c^{-1}(i)| \le \left\lceil  \frac{n}{L}\right\rceil$, the second condition of the lemma also holds for every  $\alpha \ge 1$.
  
We extend the precoloring of $G^2$ to the required coloring by the following greedy procedure: We select an arbitrary uncolored vertex $v$ and   color it by a color from  $[L]$ such that the new partial coloring also satisfies the three conditions of the lemma. In what follows, we prove that such a greedy choice of a color is always possible.

Coloring of a  vertex $v$ with a color $i$  can be {forbidden} only because it breaks one of the three conditions. Let us count, how many colors can be  forbidden for $v$ by each of the three constraints.
\begin{enumerate}
\item Vertex $v$ has at most $d^2$ neighbors in $G^2$, so the first constraint forbids at most $d^2$ colors.
\item The second constraint forbids all the colors that are ``fully packed'' already. The number of such colors is at most $\frac{n}{\left(\frac{\alpha n}{L}\right)}=\frac L\alpha$.
\item To estimate the number of colors forbidden by the third condition,  we go through all the neighbors of~$v$. A neighbor $u \in N_G(v)$ forbids a color $i$
if coloring $v$ by  $i$ exceeds the allowed bound on $k_{i,c(u)}$.
Hence to estimate the number of such forbidden colors $i$ (for every fixed vertex $u$)
we need to estimate how many  values of $k_{i,c(u)}$ can reach the allowed upper bound $K_{i,c(u)}$. We have that 
\begin{align*}
|\{i \colon k_{i,c(u)}  &=  K_{i,c(u)}\}|   \stackrel{\text{by~\eqref{eq:gh}}}{\le}  
\left| \left\{ 
i \colon k_{i,c(u)} \ge \frac{\beta n}{2L^2d^2} 
\right\} \right| = \left| \left\{ 
i \colon k_{i,c(u)} \cdot\frac{2L^2d^2}{\beta n} \ge 1
\right\} \right| & \\  &\leq  \sum_{i \in [L]}k_{i,c(u)}\cdot \frac{2L^2d^2}{\beta n}
. &
\end{align*}

The number of edges between vertices of the same color $c(u)$ and all other vertices of the graph does not exceed the cardinality of the color class $c(u)$ times $d$. Thus we have 
\begin{align*}
 \sum_{i \in [L]}k_{i,c(u)}\cdot \frac{2L^2d^2}{\beta n}
 &
 \le  d|c^{-1}(c(u))|\cdot \frac{2L^2d^2}{\beta n}    \stackrel{\text{by~\eqref{eq:balver}}}{\le} d\left\lceil\frac{\alpha n}{L}\right\rceil\cdot \frac{2L^2d^2}{\beta n} \\ &\le 
d\frac{2\alpha n}{L}\cdot \frac{2L^2d^2}{\beta n}  = 
\frac{4\alpha Ld^3}{\beta} \, .&
\end{align*}
where the last inequality is due to $\alpha>1$ and $L\le n$.

Therefore, 
\[
|\{i \colon k_{i,c(u)}  =  K_{i,c(u)}\}|  \leq \frac{4\alpha Ld^3}{\beta} \, . 
\]

Since the degree of $v$ in $G$ does not exceed $d$, we have that the  number of colors forbidden by the third constraint is at most $\frac{4\alpha Ld^4}{\beta}$.
\end{enumerate}
Thus, the total number of colors forbidden by all the three constraints for the vertex $v$ is at most
\[d^2 + \frac{L}{\alpha} + \frac{4\alpha L d^4}{\beta} \, .\]
By taking sufficiently large constants 
$\alpha$, $\beta$, and $\tau$, say $\alpha=4$, $\beta=16\alpha^2 d^4$,  and $\tau = \frac{16(d^2 +1)}{11}$, we guarantee that this expression does not exceed $L-1$
for every $L \ge \tau$. Therefore, there always exists a vacant color for the vertex $v$ which concludes the proof.
\end{proof}

Now with help of Lemma~\ref{lemma:coloring}, we describe a way to construct a specific grouping of a graph. The properties of such groupings are crucial for the final reduction. 

\begin{lemma}\label{lem:group}
For any constant $d$, there exists a constant $\lambda=\lambda(d)$ and a polynomial time algorithm
that for a given graph $G$ on $n$ vertices of maximum degree $d$ and an integer $r\leq \sqrt{\frac{n}{2\lambda}}$, 
 finds a grouping $\gr{G}$ of $G$ and a coloring $\tilde{c} \colon V(\gr{G}) \to [\lambda r]$ such that
\begin{enumerate}
\item The number of buckets of  $\gr{G}$ is 
\[|V(\gr{G})| \le \frac{|V(G)|}{r} \,;\]
\item The coloring $\tilde{c}$ is a proper coloring of $\gr{G}^2$;
\item Each bucket $B \in V(\gr{G})$ is an independent set in~$G$, i.e. for every $u,v\in B$, $uv \not\in E(G)$;
\item For every pair of  buckets $B_1,B_2 \in V(\gr{G})$ there is at most one edge between them in~$G$, i.e.
\[|\{uv \in E(G) \colon u \in B_1, v \in B_2\}| \le 1 \, .\]
\end{enumerate}
\end{lemma}

\begin{proof} Let $\beta=\beta(d)$ be a constant provided by Lemma~\ref{lemma:coloring} and let $L=\lambda r$ for $\lambda=\lambda(d)=2d\beta$.  Let also $c$ be a coloring of $G$ in $L$ colors provided by Lemma~\ref{lemma:coloring}. We want to construct a grouping $\gr{G}$ of $G$ such that for all buckets $B \in V(\gr{G})$ and all $u \neq v \in B$,
\begin{align}\label{eq:rz}
c(u)=c(v) \text{ and } c(u') \neq c(v') \\
\text{ for all } u' \in N_G(u), v' \in N_G(v).\nonumber
\end{align}
In other words, all  vertices of the same bucket are of  the same color while any two neighbors of such two vertices are of different colors.

For each color $i\in[L]$,  we introduce an auxiliary constraint graph $F_i$. The vertex set of $F_i$ is  $V(F_i)=c^{-1}(i)$ and its edge set is 
\[
E(F_i) = 
\{uv \colon \exists u'\in N_G(u),v' \in N_G(v), c(u')=c(v') \}.
\]
In our construction, each bucket of $\tilde{G}$ will be an independent set in some $F_i$. Note that this will immediately imply~(\ref{eq:rz}). The degree of any vertex $v\in V(F_i)$ is at most
\[
\deg_{F_i}(v) \leq
 \sum_{v' \in N_G(v)} (K_{c(v),c(v')}-1) \leq
 d \left(\left\lceil\frac{\beta |c^{-1}(v)|}{L}\right\rceil-1\right) \leq
 \frac{d \beta |V(F_i)|}{L} =
 \frac{|V(F_i)|}{2 r} \, .\]
This means that the greedy algorithm finds a proper coloring of each $F_i$ in at most $\frac{|V(F_i)|}{2r}+1$ colors, which splits each $F_i$ in at most $\frac{|V(F_i)|}{2r}+1$ independent sets. We create a separate bucket of $\gr{G}$ from each independent set of each $F_i$. Now we show that the four conditions from the lemma statement hold.
\begin{enumerate}
\item  For the first property, the number of independent sets in each $F_i$ is at most $\frac{|V(F_i)|}{2r}+1$. Thus the number of buckets in $\gr{G}$ is 
\[
 |V(\gr{G})|\leq
 \sum_{i\in[L]} \left(\frac{|V(F_i)|}{2r}+1\right) =
 \sum_{i\in[L]} \left(\frac{|c^{-1}(i)|}{2r}+1\right) =
 \frac{n}{2r}+L\le\frac{n}{r} \, .\]
since $L=\lambda r$ and $2\lambda r^2 \le n$.
\item For the second property, by Lemma~\ref{lemma:coloring}, the coloring $c$ is  proper in~$G^2$.
We can convert $c$ to a coloring $\tilde{c} \colon V(\gr{G}) \to [\lambda r]$ by assigning each bucket the color of its vertices (all of them have the same color). The resulting coloring $\tilde{c}$ is a proper coloring of $\gr{G}^2$ by \eqref{eq:rz} and the fact that $c$ is a proper in~$G^2$.
\item All buckets of $\gr{G}$ are monochromatic with respect to $c$, thus, each bucket $B \in V(\gr{G})$ is an independent set in~$G$ and the third property holds.
\item Finally, by (\ref{eq:rz}), there is at most one edge in $G$ between vertices corresponding to any pair of buckets in~$\gr{G}$.
\end{enumerate}
Thus, the constructed grouping and its coloring satisfy all conditions of the lemma.
\end{proof}

\subsection{Reductions}\label{sec:reductions}
\label{sec:reductions}

This section constitutes the main technical part of the paper and contains all the necessary reductions used in the lower bounds proofs. Using these reductions as building
blocks the lower bounds follow from careful calculations. The general pipeline is as follows. To prove a lower bound, we take a graph $G$ of maximum degree four that needs to be $3$-colored and construct an equisatisfiable instance $(G',H')$
of \LGH{} using Lemma~\ref{lemma:3coltolisthom}
. We then use Lemma~\ref{lemma:lhomtohom} to transform $(G',H')$ into an equisatisfiable instance $(G'',H'')$
of \GH{}. 
Thus, an algorithm checking whether there exists a homomorphism from $G''$ to $H''$ can be used to check whether the initial graph $G$ can be $3$-colored. At the same time we know a lower bound for \textsc{$3$-Coloring}
under ETH (Lemma~\ref{lemma:3col}). This gives us a lower bound for \GH{} under the ETH assumption. In order to prove the hardness of \SI{}, we show an exponential-time reduction from \GH{} to \SI{}.

\begin{lemma}[\textsc{3-Coloring} $\to$ \LGH{}]\label{lem:coltolhom}
There exists an algorithm that takes as input a graph $G$ on $n$ vertices of maximum degree $d$ 
that needs to be $3$-colored and an integer $r=o(\sqrt{n})$ and finds an equisatisfiable instance $(G',H')$ of LIST-HOM, where
$|V(G')|\le n/r$  and  $|V(H')| \le \gamma(d)^{r}$, where $\gamma(d)$ is a function of the graph degree. The running time of the algorithm is polynomial
in $n$ and the size of the output graphs.
\label{lemma:3coltolisthom}
\end{lemma}
\begin{proof}
{\em Constructing the graph $G'$.}
Let $G'$ be the grouping of $G$ and $c\colon V(G') \to [L]$ be the coloring provided by Lemma~\ref{lem:group} where~$L=\lambda(d) r$. To distinguish colorings of $G$ and $G'$, 
we call $c(B)$, for a bucket $B \in V(G')$, a~\emph{label} of~$B$.
Consider a bucket $B \in V(G')$, i.e., a subset of vertices of~$G$, and a label $i \in [L]$. From item 2 of Lemma~\ref{lem:group} we know that $c$ is a proper coloring of $(G')^2$. This, in particular, means that there is at most one $B' \in N_{G'}(B)$ such that $c(B')=i$. Moreover, if such $B'$ exists then, by item 4 of Lemma~\ref{lem:group}, there exists a unique $u \in B$ and unique $u' \in B'$ such that $uu' \in E(G)$. This allows us to define the following mapping $\phi_B \colon [L] \to B \cup \{0\}$: $\phi_B(i)=u$ if such $B'$ exists and $\phi_B(i)=0$ if $B$ has no neighbor $B'$ of label~$i$. Without loss of generality we assume that $G$ does not have isolated vertices. Since each vertex has a neighbor outside of its bucket (it cannot have a neighbor in its own bucket as buckets are independent), $B \subseteq \phi_B(L)$.

{\em Constructing the graph $H'$.}
We now define a redundant encoding of a $3$-coloring of a bucket $B \in V(G')$. Namely, let $\mu_B \colon (f \colon B \to \{1,2,3\}) \to \{0,1,2,3\}^L$.
That is, for a $3$-coloring $f \colon B \to \{1,2,3\}$ of $B$, $\mu_B$ is a vector $v$ of length~$L$. For $i \in [L]$, by $v[i]$ we denote the $i$-th component of~$v$. The value of $v[i]$ is defined as follows:
if $\phi_B(i)=0$ then $v[i]=0$, otherwise $v[i]=f(\phi_B(i))$.
In~other words, for a given bucket $B$ and a $3$-coloring $f$ of its vertices, for each possible label~$i \in [L]$, $\mu_B$
is the color of the vertex $u \in B$ that has a neighbor in a bucket with label $i$, and $0$ if there is no such vertex~$u$. 
  
We are now ready to construct the graph $H'$.
The set of vertices of $H'$ is defined as follows:
\[ V(H') =\{(R,l) \colon  R \in \{0,1,2,3\}^L \text{ and } l \in [L] \}\,,\] i.e., a~vertex of $H'$ is an encoding of a $3$-coloring of a bucket and a label of a bucket. The list constraints of this instance of \LGH{} are defined as follows:
a~bucket $B \in V(G')$ is allowed to be mapped to $(R,l) \in V(H')$ if and only if $l=c(B)$ and there is a $3$-coloring $f$ of $B$ such that $\mu_B(f)=R$.
Informally, two vertices in $V(H')$ are joined by an edge 
if they define two consistent $3$-colorings. Formally, 
$(R_1,l_1)(R_2, l_2) \in E(H')$ if and only if
$R_1[l_2] \neq R_2[l_1]$. Note that $|V(G')| \le n/r$ by Lemma~\ref{lem:group} and
$|V(H')| \le 4^L \cdot L \le 4^L \cdot 2^L=8^{\lambda(d)r}=\gamma(d)^r$
for $\gamma(d)=8^{\lambda(d)}$. 

{\em Running time of the reduction.} The reduction clearly takes time polynomial in the size of input and output.

{\em Correctness of the reduction.} It remains to show that $G$ is $3$-colorable if and only if $(G',H')$ is a yes-instance of \LGH.

Assume that $G$ is $3$-colorable and take a proper $3$-coloring $g$ of~$G$. It defines a homomorphism from $G'$ to $H'$ in a natural way: $B \in V(G')$ is mapped to $(\mu_B(g|_B), c(B))$, where $g|_B$
is the function $g$ with its domain restricted to $B$.
Each list constraint is satisfied   by definition. To show that each edge is mapped to an edge, consider an edge $BB' \in E(G')$. Then, by item 4 of Lemma~\ref{lem:group} there is a unique edge $uu' \in E(G)$ such that $u \in B, u' \in B'$.
Note that $B$ and $B'$ are mapped to vertices $(R,l)$ and $(R',l')$ such that $R[l']=g(u)$ and $R'[l]=g(u')$. Since $g$
is a proper $3$-coloring of $G$, $g(u) \neq g(u')$. This, in turn, means that $(R,l)(R',l') \in E(H')$ and hence the edge $BB'$ is mapped to this edge in~$H'$.

For the reverse direction, consider a homomorphism $h \colon G' \to H'$. For each bucket $B \in V(G')$, $h(B)$ defines a proper $3$-coloring of~$B$. Together, they define a $3$-coloring $g$ of $G$ and we need to show that $g$ is proper. Assume, to the contrary, that there is an edge $uu' \in E(G)$ such that $g(u)=g(u')$. By item 3 of Lemma~\ref{lem:group}, $u$ and $u'$
belong to different buckets $B, B' \in V(G')$. By the definition of grouping, $BB' \in E(G')$. Since $h$ is a homomorphism,
$(R,l)(R',l') := h(B)h(B') \in E(H')$. At the same time, $R[l']=g(u)=g(u')=R'[l]$ which contradicts the fact that 
$(R,l)(R',l')$ is an edge in~$H'$.
\end{proof}

\begin{lemma}[\LGH $\to$ \GH]
\label{lemma:lhomtohom}
There is a polynomial-time algorithm that from an instance 
$(G,H)$ of LIST-HOM where $|V(G)|=n$, $|V(H)|=h$ constructs an equisatisfiable instance $(G',H')$
of HOM where $|V(G')| \le n+\Delta$, $|V(H')| \le \Delta$ for $\Delta=25h^2$.
\end{lemma}
\begin{proof}
{\em Preparations.}
We start with a simple $6$-vertex gadget $D'$ consisting of a $5$-cycle together with an apex vertex adjacent to all the vertices of the cycle, see Fig.~\ref{fig:D}. 

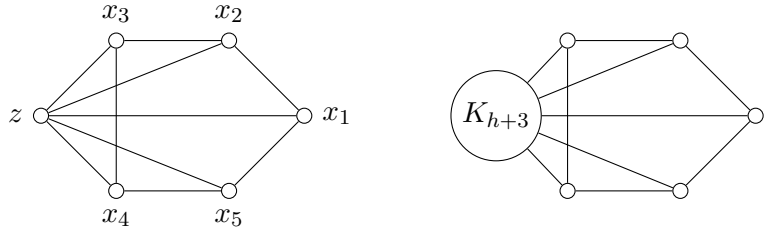
\begin{figure}[ht]
\begin{center}
\begin{tikzpicture}
\begin{scope}[xshift=-6cm]
\foreach \x/\y/\n/\w in {2/0/1/0, 1/1/2/90, -0.5/1/3/90, -0.5/-1/4/-90, 1/-1/5/-90}
\node[circle,draw,minimum size=2mm,inner sep=0mm,                                                 
  label=\w:{$x_{\n}$}] (b\n) at (\x,\y) {};                                                         
\draw (b1) -- (b2) -- (b3) -- (b4) -- (b5) -- (b1);                                                 
\node[circle,draw,minimum size=2mm,inner sep=0mm,label=180:$z$] (z) at (-1.5,0) {};                 
\foreach \n in {1,2,...,5} 
\draw (b\n) -- (z);                    
\end{scope}

\foreach \xs in {0} {

  \begin{scope}%
    \foreach \x/\y/\n in {2/0/1, 1/1/2, -0.5/1/3,-0.5/-1/4,1/-1/5}
  \node[circle,draw,minimum size=2mm,inner sep=0mm] (b\n) at (\x+\xs,\y) {};
  \draw (b1) -- (b2) -- (b3) -- (b4) -- (b5) -- (b1);
  \node[circle,draw,minimum size=12mm,inner sep=0mm] (z) at (-1.45+\xs,0) {$K_{h+3}$};
  \foreach \n in {1,2,...,5}
  \draw (b\n) -- (z);
  \end{scope}
}
\end{tikzpicture}
\end{center}
\caption{The graphs~$D'$ and $D$. The encircled clique $K_{h+3}$ is the canonical clique of $D$. 
  An edge from a clique to a vertex of a cycle means that each vertex of the clique is joined to this vertex.}\label{fig:D}
\end{figure}

An important property of $D'$ is that for each homomorphism $\phi  \colon D' \to D'$ and $i \in [5]$,
\[ \phi(z) =z \text{ and }\phi(z) \neq \phi(x_i).\]
In words, $z$ is always mapped to $z$ and nothing else is mapped to~$z$. Indeed, because the vertex $z$ is adjacent to all the remaining vertices of~$D'$, we have that $\phi(z) \neq \phi(x_i)$. By the same reason, we have that for every $i\in [5]$, $\phi(x_i)\in N_{D'}(\phi(z))$.  But for every $x_i$ its open neighborhood $N_{D'}(x_i)$ induces a bipartite graph. On the other hand, the chromatic number of the cycle $C=x_1x_2x_3x_4x_5$   is three, and thus it cannot be mapped by $\phi$ to 
$N_{D'}(x_i)$ for any $i\in [5]$. Therefore, $\phi(z) =z$.

In order for the $\phi(z)=z$ argument to work in a bigger graph,
we replace $z$ by a clique $K_{h+3}$ of size $h+3$, called
the {\em canonical clique} of the gadget.
The obtained graph with $(h+3)+5$ vertices is denoted as $D$ (see Fig.~\ref{fig:D}).

Let $D_0, \ldots, D_{k}$ be $k+1$ copies of the graph $D$.
We join those $k+1$ graphs isomorphic to $D$ to construct a larger gadget $T_k$ 
as follows (see Fig.~\ref{fig:T_kh3}).
For each $i \in [k]$, we select an arbitrary vertex from the canonical
clique of $D_i$, denote this vertex as $z_i$, and identify
it with one arbitrary vertex of $D_{i-1}$ which does not belong
to the canonical clique of $D_{i-1}$, i.e., with a vertex
of the $5$-cycle.
and connect every vertex of $K_{h+3}$ to all neighbors of $z$ in the subsequent block. 
We also mark one of $K_{h+3}$'s vertices as $z$ and connect it to the left of it, see Fig.~\ref{fig:T_kh3}. Denote the new graph by $T_{k,h+3}$.
Oberve that each $D_i$ is a block of $T_k$ and we call $D_i$ the $i$th block of $T_k$.
Note that two consecutive blocks $D_{i-1}$ and $D_i$ have exactly one common vertex, namely $z_i$.

The reason we are using those canonical cliques instead of single vertices
in the construction of $T_k$ is that those canonical cliques 
are big enough to behave as anchors.
That is, we will prove that canonical cliques can only be mapped to themselves and
not to other parts of the graph, in particular, for each $i \in [k]$
and homomorphism $\phi \colon T_k \to T_k$, $\phi(z_i)=z_i$.

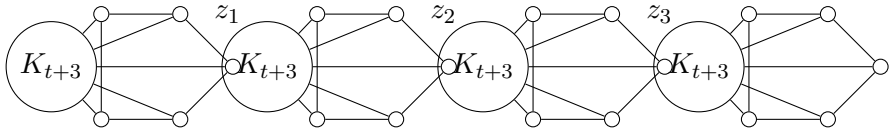
\begin{figure}[ht]
\begin{center}
\begin{tikzpicture}[scale=0.7]
  \foreach \xs in {0, 4.1, 8.2, 12.3} {

  \begin{scope}%
  \foreach \x/\y/\n in {2/0/1, 1/1/2, -0.5/1/3,-0.5/-1/4,1/-1/5}
    \node[circle,draw,minimum size=2mm,inner sep=0mm] (b\n) at (\x+\xs,\y) {};
  \draw (b1) -- (b2) -- (b3) -- (b4) -- (b5) -- (b1);
  \node[circle,draw,minimum size=12mm,inner sep=0mm] (z) at (-1.45+\xs,0) {$K_{t+3}$};
  \foreach \n in {1,2,...,5}
    \draw (b\n) -- (z);
  \end{scope}
  }
  
  \foreach \x/\n in {2/1, 6.1/2, 10.2/3}
    \node[anchor=south,circle,inner sep=2mm] at (\x-.1,.3) {$z_{\n}$};
  
\end{tikzpicture}
\end{center}
\caption{The gadget $T_{k}$ for $k=3$.  }\label{fig:T_kh3}
\end{figure}

{\em Constructing $G'$.}
Let $A_h$ be a graph consisting of a matching with $h$ edges  $\{ a_1b_1,\ldots, a_hb_h\}$.
Then the graph $G'$  consists of a copy of $G$, a copy of $T_{h}$, and a copy of $A_h$ with the following additional edges: the vertex $z_i$ from the $i$th block of $T_{h}$ is adjacent to the vertices $a_i$ and $b_i$. Also we add edges from $G$ to $A_h$: for a vertex $g_i\in G$ we add an edge $g_ia_j$ for every $j$, and an edge $g_ib_j$ if $j\not\in\mathcal{L}(i)$ (see Fig.~\ref{fig:gprime}). The number of vertices
in $G'$ is at most $n+2h+(h+1)(h+3+5) \le n+(h+1)(h+11) \le n+25h^2$.

\begin{figure}[ht]
\begin{center}
\begin{tikzpicture}[scale=0.7]
  \foreach \xs/\n in {0/1, 4.1/2, 8.2/3} {

  \begin{scope}%
  \foreach \x/\y/\n in {2/0/1, 1/1/2, -0.5/1/3,-0.5/-1/4,1/-1/5}
    \node[circle,draw,minimum size=2mm,inner sep=0mm] (b\n) at (\x+\xs,\y) {};
  \draw (b1) -- (b2) -- (b3) -- (b4) -- (b5) -- (b1);
  \node[circle,draw,minimum size=12mm,inner sep=0mm] (z) at (-1.45+\xs,0) {$K_{h+3}$};
  \foreach \n in {1,2,...,5}
    \draw (b\n) -- (z);
    
  \node[circle,draw,minimum size=2mm,inner sep=0mm,label=left:$a_{\n}$] (p\n) at (1.7+\xs,-2) {};
  \node[circle,draw,minimum size=2mm,inner sep=0mm,label=right:$b_{\n}$] (q\n) at (2.3+\xs,-2) {};
  \draw (b1) -- (p\n) -- (q\n) -- (b1);
  \end{scope}
  }
  
  \foreach \xs in {12.3} {

  \begin{scope}%
  \foreach \x/\y/\n in {2/0/1, 1/1/2, -0.5/1/3,-0.5/-1/4,1/-1/5}
    \node[circle,draw,minimum size=2mm,inner sep=0mm] (b\n) at (\x+\xs,\y) {};
  \draw (b1) -- (b2) -- (b3) -- (b4) -- (b5) -- (b1);
  \node[circle,draw,minimum size=12mm,inner sep=0mm] (z) at (-1.45+\xs,0) {$K_{h+3}$};
  \foreach \n in {1,2,...,5}
    \draw (b\n) -- (z);
  \end{scope}
  }
  
  \foreach \x/\n in {2/1, 6.1/2, 10.2/3}
    \node[anchor=south,circle,inner sep=2mm] at (\x,.3) {$z_{\n}$};
    
  \node[circle,draw,minimum size=20mm] (g) at (6.1,-5) {$G$};
  \node[circle,draw,minimum size=2mm,inner sep=0mm,label=left:$i$] (i) at (6.1,-4.4) {};
  
  \draw (i) -- (p1);
  \draw (i) -- (q1);
  \draw (i) -- (p2);
  \draw (i) -- (p3);
 
\end{tikzpicture}
\end{center}
\caption{The graph~$G'$. A vertex $i \in V(G)$ is connected to $b_j$ if and only if $j \not \in \cL(i)$, where $\cL(i)$ is the list associated with the vertex $i\in V(G).$}\label{fig:gprime}
\end{figure}
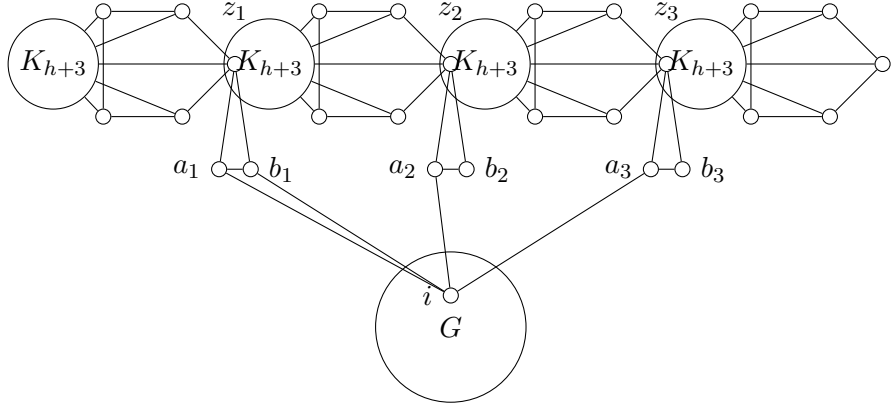

{\em Constructing $H'$.}
The graph $H'$ is constructed similarly. It consists of a copy of $H$, a copy of $T_{h}$, and a copy of $A_h$. For every $i$ we add edges $z_ia_i$ and $z_ib_i$ as before. Also, each vertex $i$ of $H$ is adjacent to all the vertices from $A_h$ except for $b_i$ (see Fig.~\ref{fig:hprime}).
The number of vertices
in $H'$ is at most $h+2h+(h+1)(h+3+5) \le (h+1)(h+11) \le 25h^2$.

\begin{figure}[ht]
\begin{center}
\begin{tikzpicture}[scale=.7]
  \foreach \xs/\n in {0/1, 4.1/2, 8.2/3} {

  \begin{scope}%
  \foreach \x/\y/\n in {2/0/1, 1/1/2, -0.5/1/3,-0.5/-1/4,1/-1/5}
    \node[circle,draw,minimum size=2mm,inner sep=0mm] (b\n) at (\x+\xs,\y) {};
  \draw (b1) -- (b2) -- (b3) -- (b4) -- (b5) -- (b1);
  \node[circle,draw,minimum size=12mm,inner sep=0mm] (z) at (-1.45+\xs,0) {$K_{h+3}$};
  \foreach \n in {1,2,...,5}
    \draw (b\n) -- (z);
    
  \node[circle,draw,minimum size=2mm,inner sep=0mm,label=left:$a_{\n}$] (p\n) at (1.7+\xs,-2) {};
  \node[circle,draw,minimum size=2mm,inner sep=0mm,label=right:$b_{\n}$] (q\n) at (2.3+\xs,-2) {};
  \draw (b1) -- (p\n) -- (q\n) -- (b1);
  \end{scope}
  }
  
  \foreach \xs in {12.3} {

  \begin{scope}%
  \foreach \x/\y/\n in {2/0/1, 1/1/2, -0.5/1/3,-0.5/-1/4,1/-1/5}
    \node[circle,draw,minimum size=2mm,inner sep=0mm] (b\n) at (\x+\xs,\y) {};
  \draw (b1) -- (b2) -- (b3) -- (b4) -- (b5) -- (b1);
  \node[circle,draw,minimum size=12mm,inner sep=0mm] (z) at (-1.45+\xs,0) {$K_{h+3}$};
  \foreach \n in {1,2,...,5}
    \draw (b\n) -- (z);
  \end{scope}
  }
  
  \foreach \x/\n in {2/1, 6.1/2, 10.2/3}
    \node[anchor=south,circle,inner sep=2mm] at (\x,.3) {$z_{\n}$};
    
  \node[circle,draw,minimum size=20mm] (g) at (6.1,-5) {$H$};
  \node[circle,draw,minimum size=2mm,inner sep=0mm,label=left:$i$] (i) at (6.1,-4.4) {};
  
  \draw (i) -- (p1);
  \draw (i) -- (q2);
  \draw (i) -- (p2);
  \draw (i) -- (p3);
  \draw (i) -- (q3);
  
\end{tikzpicture}
\end{center}
\caption{The graph~$H'$. A vertex $i \in V(H)$ is connected to all $a_j$'s and all $b_j$'s except for $b_i$.}
\label{fig:hprime}
\end{figure}
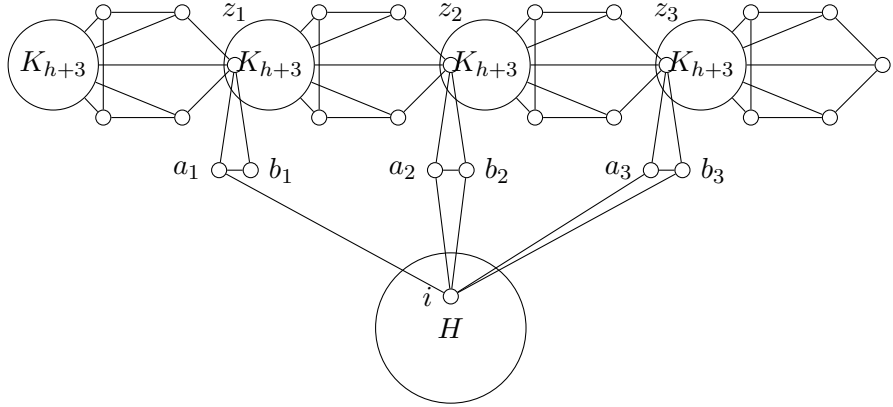

{\em Correctness.} We now turn to prove that the instance $(G,H)$ of LIST-HOM is equisatisfiable to an instance $(G',H')$
of HOM.

\begin{claim}
Any homomorphism $\phi$ from $G'$ to $H'$ maps $T_{h}$ into $T_{h}$.
\end{claim}
\begin{proof}[Proof of the claim]
No pair of vertices of the same clique of $T_{h}$ is mapped to the same vertex in $H'$, because $H'$ has no self-loops. Therefore, canonical cliques from $T_{h}$ 
are mapped to some cliques from $T_{h}$, as $H'$ has no more cliques of size $h+3$. The remaining vertices of $T_{h}$ have at least $h+3$ neighbors from 
canonical cliques, therefore they must be mapped to vertices from $T_{h}$.
\end{proof}

\begin{claim}
Any homomorphism $\phi$ from $G'$ to $H'$ bijectively maps $T_{h}$ to $T_{h}$ so that the order of $z$'s is preserved, i.e., for 
each $i \in [k]$, $\phi(z_i) = z_i$.
\end{claim}
\begin{proof}[Proof of the claim]
\begin{enumerate}
\item {\em Every canonical clique is mapped to a canonical clique.}
First note that a canonical clique is mapped into one block. Indeed, there are no vertices outside of a block that are connected to more than one vertex of the block. Assume, to the contrary, that 
same canonical clique is mapped to one block but not to a canonical clique. Then its image has to contain one or two vertices of the $5$-cycle from that block. If the image contains only one vertex of the $5$-cycle, then the image of the $5$-cycle has at most $3$ vertices: one vertex from the canonical clique $K_{h+3}$, two neigbors of the vertex from the $5$-cycle (because all the vertices of the image of the $5$-cycle must be connected to all the vertices of the image of the clique). Note that these three vertices do not form a triangle, therefore the $5$-cycle cannot be mapped to them. If the image of the clique contains two vertices outside of the canonical clique, then for the same reason the image of the $5$-cycle must contain only two vertices, which is not possible. 
This analysis shows that every canonical clique $K_{h+3}$ must be mapped to a canonical clique $K_{h+3}$.
\item {\em Every block is mapped to a block.}  We already know that every canonical clique is mapped to a canonical clique. The $5$-cycle from the same block must be mapped to the corresponding $5$-cycle, because it is the only image that contains a closed walk of odd length and every vertex of which is connected to the clique (recall that the images of the canonical clique and the $5$-cycle do not intersect, since their preimages are joined by edges). Note that since the canonical clique and the cycle are mapped to themselves, $z_i$ has to be mapped to some $z_j$.
\item {\em If $D_i$ is mapped to $D_j$, then $D_{i+1}$ is mapped to $D_{j+1}$.}
The cycle from $D_i$ shares a vertex with the canonical clique from $D_{i+1}$, therefore if $D_i$ is mapped to $D_j$, then
$D_{i+1}$ can only be mapped to $D_{j+1}$ or $D_j$. However, $D_{i+1}$ cannot be mapped into the same block as $D_i$. 
Indeed, in this case the canonical clique of $D_{i+1}$ would be mapped to the canonical clique of $D_j$,
but we already know that $z_{i+1}$ is mapped to the $5$-cycle of $D_j$.
Therefore, $D_i$ and $D_{i+1}$ must be mapped in consecutive blocks.
\end{enumerate}
The above proves that for every $i \in \{0,\ldots,k\}$, $D_i$ is mapped to $D_i$, which implies that any homomorphism preserves the order of $z$'s.
\end{proof}

\begin{claim}
Any homomorphism $\phi$ from $G'$ to $H'$ maps $A_h$ to $A_h$ so that $a_ib_i$ is mapped to $a_ib_i$.
\end{claim}
\begin{proof}[Proof of the claim]
Every pair $a_ib_i$ is connected to $z_i\in T_{h}$, so it can be mapped either to $a_ib_i$ or to some vertices of $T_{h}$. But in the latter case it would not have paths of length $2$ to all other pairs $a_jb_j$.
\end{proof}

\begin{claim}\label{claim4}
Any homomorphism $\phi$ from $G'$ to $H'$ maps $G$ to $H$.
\end{claim}
\begin{proof}[Proof of the claim]
Assume, to the contrary, that a vertex $g\in V(G)$ is mapped to a vertex $v\in V(T_{h})$ or a vertex $a\in V(A_h)$.
The vertex $g$ is adjacent to at least $h$ vertices from $A_h$, but $v$ and $a$ are adjacent to at most $2$ vertices from $A_h$  (recall that by the previous claim every $a_ib_i$ is mapped to $a_ib_i$).
\end{proof}

Now we show that the two instances are equisatisfiable. Let $\phi$ be a 
list homomorphism from $G$ to $H$. We show that its natural extension $\phi'$
mapping $T_{h}$ to $T_{h}$ and $A_h$ to $A_h$ is a correct homomorphism from $G'$ to $H'$.
This is non-trivial only for edges of $G'$ from $G$ to $A_h$. Consider an edge from
a vertex $i$ of $G$ to a vertex $b_j$. The presence of this edge means that $i$
is not mapped to $j$ by $\phi$. Recall that the $b_j$ is mapped by $\phi$ to $b_j$. This means that the considered edge in $G'$ is mapped to an edge in $H'$ by $\phi'$.

For the reverse direction, let $\phi'$ be a homomorphism from $G'$ to $H'$.
We show that its natural projection is a list homomorphism from $G$ to $H$.
Since $\phi'$ maps $G$ to $H$ (by Claim~\ref{claim4}) it is enough to check that all list constrains are satisfied.
For this, consider a vertex $i$ from $G$ and assume that $j \not \in \mathcal{L}(i)$. Then $\phi'$ does not map $i$ to $j$ as otherwise
there would be no image for one of the edges $g_ia_j$ or $g_ib_j$, where $g_i$ is the $i$th vertex of $G$.

{\em Running time of the reduction.} The reduction clearly takes time polynomial in the input length.
\end{proof}

\section{Lower Bounds}\label{sec:lower}
\subsection{Graph Homomorphism}\label{sec:hom}
We are ready to prove our main result about graph homomorphisms, i.e., Theorem~\ref{main:theorem_homs}.

\begin{theorem}[Theorem~\ref{main:theorem_homs} restated]\label{thm:main}
Let $G$ be an $n$-vertex graph and $H$ be a graph with at most $h:=h(n)$ vertices.
Unless ETH fails, for any constant $D\ge1$ there exists a constant $c=c(D)>0$ such that
for any non-decreasing function $3\le h(n)\le n^D$, 
there is no $\cO\left(h^{cn}\right)$ time algorithm deciding whether there is a homomorphism from $G$ to $H$. 
\end{theorem}

\begin{proof}
The outline of the proof of the theorem is as follows. Assuming that there is a ``fast" algorithm for \GH, we show that there is  also a  ``fast"  algorithm solving \LGH, which, in turn, implies  ``fast" algorithm for \textsc{3-Coloring} on degree $4$ graphs, contradicting ETH. In what follows, we specify what we mean by ``fast".

Let $h_0=25^2$. If $h(n)<h_0$ for all values of $n$, then an algorithm with running time $\cO\left(h^{cn}\right)$ would solve \textsc{3-Coloring} in time $\cO\left(h_0^{cn}\right)=\cO\left(2^{cn\log{h_0}}\right)$ (recall that $h(n) \ge 3$). Therefore, by choosing a small enough constant $c$ such that $c\log{h_0} < q$, we arrive to a contradiction with Lemma~\ref{lemma:3col}.

From now on we assume that $h(n)\ge h_0$ for large enough values of $n$.
Let $c=\frac{q}{8D\log{\gamma}}$, where $q$ is the constant from Lemma~\ref{lemma:3col}, and $\gamma:=\gamma(4)$ is the constant from Lemma~\ref{lem:coltolhom}. For 
  the sake of contradiction, let us assume that there exists an algorithm $\cal A$ deciding whether  $G\to H$ in time $\cO(h^{cn})=\cO(2^{cn\log{h}})$, where $|V(G)|=n, |V(H)|=h:=h(n)$. Now we show how to solve $3$-coloring on $n'$-vertex graphs of maximum degree four in time $2^{q n'}$, which would contradict  Lemma~\ref{lemma:3col}.
 
Let $G'$ be an $n'$-vertex graph of maximum degree four that needs to be $3$-colored. 
Let $r=\frac{\log{h}}{4D\log{\gamma}}$ and $n=\frac{2n'}{r}$. 
Using Lemma~\ref{lem:coltolhom} (note that $r=o(\sqrt{n'})$ as required) we construct an instance $(G_1,H_1)$ of \LGH{} that is satisfiable if and only if the initial graph $G'$ is $3$-colorable, and $|V(G_1)|\le\frac{n'}{r}, |V(H_1)|\le\gamma^r$. By Lemma~\ref{lemma:lhomtohom}, this instance is equisatisfiable to an instance $(G,H)$ of \GH{} where $|V(H)|<25\gamma^{2r}= 25h^{\frac{1}{2D}}\le h$ (since $D\ge1$ and $h(n)\ge h_0$), and 
\[|V(G)| \le \frac{n'}{r}+25\gamma^{2r} \le \frac{n}{2}+25h^{\frac{1}{2D}}\le \frac{n}{2}+25\sqrt{n}\le n  \]
(for sufficiently large values of~$n$).

Now, in order to solve $3$-coloring for $G'$, we construct an instance $(G,H)$ with $|V(G)|\le n$ and $|V(H)|\le h$ of 
\GH \,  and 
invoke the algorithm $\cal A$ on this instance.
 The running time of $\cal A$ is
\[
\cO( 2^{cn\log{h}}) =
\cO(2^{\frac{2cn'}{r}\log{h}} )=
 \cO(2^{2cn'\log{h}\cdot\frac{4D\log{\gamma}}{\log{h}}} )=
 \cO(2^{8cDn'\log{\gamma}} )=
 \cO(2^{q n'})
 \, \]
and hence we can find a 3-coloring of $G'$ in time $\cO(2^{q n'})$, which contradicts ETH (see Lemma~\ref{lemma:3col}). 
\end{proof}

\begin{theorem}[Theorem~\ref{main:theorem_local_homs} restated]
Let $G$ be an $n$-vertex graph $G$ and $H$ be a graph with at most $h:=h(n)$ vertices.
Unless ETH fails, for any constant $D\ge1$ there exists a constant $c=c(D)>0$ such that
for any non-decreasing function
$3\le h(n)\le n^D$, there is no $\cO\left(h^{cn}\right)$ time algorithm deciding whether there is a locally injective homomorphism from $G$ to $H$. 
\end{theorem}

\begin{proof}
The proof is almost identical to the proof of Theorem~\ref{thm:main}.

Let us  observe that in the reduction  
in Lemma~\ref{lemma:3coltolisthom},  in graph $G'$, we take a coloring (in the proof we refer to such coloring as to labeling) of the square of $G'$. Thus 
for every bucket $v$ of $G'$, all its neighbors are  labeled by different colors. The way we construct the lists, 
only buckets with the same labels can be mapped to the same vertex of $H'$. Thus 
for every vertex $v$ of $G'$, no pair of its neighbors can be mapped to the same vertex. Hence every list homomorphism from $G'$ to $H'$ is locally injective. Therefore the result of Lemma~\ref{lemma:3coltolisthom} holds for locally injective list homomorphisms as well and we obtain the following lemma.

\begin{lemma}\label{lemma:3-local-col-list}
There exists an algorithm that takes as input a graph $G$ on $n$ vertices of maximum degree $d$ 
that needs to be $3$-colored and an integer $r=o(\sqrt{n})$ and finds an equisatisfiable instance $(G',H')$ of \LIGH, where
$|V(G')|\le n/r$  and  $|V(H')| \le \gamma(d)^{r}$, where $\gamma(d)$ is a function of the graph degree. The running time of the algorithm is polynomial
in $n$ and the size of the output graphs.\end{lemma}

In the reduction of  Lemma~\ref{lemma:lhomtohom}, we established that every homomorphism from $G'$ to $H'$
maps   $T_{h}$ to $T_{h}$ and $A_h$ to $A_h$ so that $a_ib_i$ is mapped to $a_ib_i$. 
Thus for vertices of these structures, every homomorphism is locally injective. 
By Claim~\ref{claim4}, 
any homomorphism $\phi$ from $G'$ to $H'$ maps $G$ to $H$. Therefore there  is a locally injective homomorphism from  $G'$ to $H'$ if and only if there is a  locally injective list homomorphism from $G$ to $H$.  Then by making use of Lemma~\ref{lemma:3-local-col-list},  the calculations performed in the proof of Theorem~\ref{thm:main} we conclude with the proof of the theorem.
\end{proof}

\subsection{Subgraph Isomorphism}\label{sec:subiso}

To prove a lower bound for \SI we need a   reduction,
which given an instance of \GH 
produces a single exponential number of instances of \SI.
Even though from the perspective of polynomial time algorithms
such a reduction gives no implication in terms of which problem is harder,
in our setting it is enough to obtain a lower bound for \SI.

\begin{theorem}\label{thm:homtoiso}
Given an instance $(G, H)$ of \GH{} one can in $\operatorname{poly}(n)2^n$ time create $2^n$ instances of \SI{} with $n$ vertices, where $n = |V(G)|+|V(H)|$, such that $(G, H)$ is a yes-instance if and only if at least one of the created instances of \SI{}  is a yes-instance.
\end{theorem}
\begin{proof}
Let $(G,H)$ be an instance of \GH{} and let $n = V(G) + V(H)$.
Note that any homomorphism $h$ from $G$ to $H$
can be associated with some sequence of non-negative
numbers $(|h^{-1}(v)|)_{v \in V(H)}$, being
the numbers of vertices of $G$ mapped to particular 
vertices of $H$.
The sum of the numbers in such a sequence equals exactly $|V(G)|$.
As the number of such sequences is $\binom{V(G) + V(H)-1}{V(H)-1} \le 2^{n}$,
we can enumerate all such sequences in time $2^n\operatorname{poly}(n)$.
For each such sequence $(a_v)_{v \in V(H)}$
we create a new instance $(G',H')$ of \SI,
where the pattern graph remains the same, i.e., $G' = G$,
and in the host graph $H'$ each vertex of $v \in V(H)$
is replicated exactly $a_v$ times (possibly zero).
Observe that $|V(H')| = |V(G')|$.

We claim that $G$ admits a homomorphism to $H$ if and only if
for some sequence $(a_v)_{v \in V(H)}$ the
graph $G'$ is a subgraph of $H'$.
First, assume that $G$ admits a homomorphism $h$ to $H$.
Consider the instance $(G',H')$ created
for the sequence $a_v = |h^{-1}(v)|$
and observe that we can create a bijection $h' : V(G') \to V(H')$
by assigning $v \in V(G')$ to its private copy of $h(v)$.
As $h$ is a homomorphism, so is $h'$, and as $h'$ is at the same
time a bijection, we infer that $G'$ is a subgraph of $H'$.

On the other hand if for some sequence $(a_v)_{v\in V(H)}$
the constructed graph $G'$ is a subgraph of $H'$,
then projecting the witnessing injection $g : V(G') \to V(H')$
so that $g'(v)$ is defined as the prototype of the copy $g(v)$
gives a homomorphism from $G$ to $H$,
as copies of each $v \in V(H)$ form independent sets in $H'$.
\end{proof}

 Combining Theorem~\ref{thm:main} with Theorem~\ref{thm:homtoiso}, we immediately obtain the following lower bound.
\begin{theorem}\label{thmSI}
Unless ETH fails, there exists a constant $c>0$ such that there is no algorithm deciding whether a given $n$-vertex graph $G$ contains a subgraph isomorphic to a given $n$-vertex graph $H$ in time 
$\cO\left(n^{cn}\right)$.
\end{theorem}

\section{Conclusion and Open Problems}\label{sec:conop}
In this work we resolved a number of  questions about exact exponential algorithms. Our lower bounds suggest   several directions for further research.  
\paragraph{``Fine-grained" dichotomy}
The classical results of ~\citeN{HellN90-On} establishes the following dichotomy for \GH  subject to  P$\neq$ NP  :  For every fixed simple graph $H$,
the problem whether there exists a homomorphism from $G$ to $H$  is solvable in
polynomial time if and only if $H$ is bipartite. 
Is there anything similar to that in the world of exponential algorithms for  HOM$(G,H)$? 

More precisely,  for graph classes $\mathcal{G}$ and $\mathcal{H}$ we denote
by HOM$(\mathcal{G}, \mathcal{H})$ the restriction of the graph homomorphism problem
to input graphs $G\in \mathcal{G}$ and $H\in\mathcal{H}$. If $\mathcal{G}$ or $\mathcal{H}$ is the class of all graphs then we use the placeholder `$\_$' instead of a letter.
Thus the result of  Hell-Ne\v set\v ril states that unless P$\neq$NP, HOM$(\_,\mathcal{H})$ is in P if and only if 
$\mathcal{H}$ is a class of bipartite graphs.

 Now we know that
solving   HOM$(\_,\_)$ with input graphs $G$ and $H$ in time 
$|V(H)|^{o(|V(G)|)}$ would refute ETH. On the other hand, when $ \mathcal{H}$ is the class of graphs consisting of complete graphs, HOM$(\_,\mathcal{H})$ is equivalent to computing the chromatic number of   $G$ and thus is solvable in time $\cO(2^{|V(G)|})$ \cite{BjorklundHK2009-Se}. 
More generally, let  $ \mathcal{H}$ be a graph class  such that for some constant $t$, either the clique-width or the maximum vertex degree of the core of every graph in 
$  \mathcal{H}$ is at most $t$.  \citeN{Wahlst10} have shown that  in this case 
HOM$(\_,\mathcal{H})$ is solvable in single-exponential time  $\cO(f(t)^{|V(G)|})=2^{\cO(|V(G)|)}$, where $f$ is some function of $\mathcal{H}$ only.  Is it possible  to  characterize (up to some complexity assumption)    graph classes $\mathcal{H}$, where HOM$(\_,\mathcal{H})$ is solvable in single-exponential time?


What about the fine-grained complexity of \GH   for 
HOM$(\mathcal{G}, \_)$ and   HOM$(\mathcal{G},\mathcal{H})$? 
Of course, similar questions are interesting for \SI,  as well as for 
 counting versions of \GH and \SI.

\paragraph{Some concrete problems}   Are the following problems solvable in single-exponential time?
 \begin{itemize}
 \item   \SI  with instance $(G,H)$  when the maximum vertex degree of $G$ is $3$. (When degree of $G$ does not exceed $2$, the problem is solvable in single-exponential time, see e.g. \cite{HeldKarp62}.)
\item 
  Deciding if graph $G$ can be obtained from graph $H$ only by edge-contractions.
  \item 
 Deciding if graph $G$ is an immersion of graph $H$.
 \item Deciding if  $G$ is a minor of a graph $H$ for the special situation when $G$ is a clique.
 \item  Finding a minimum distortion embedding into a cycle. We remark  that  embedding in a path can be done in time $2^{\cO(|V(G)|)}$ \cite{bw,FominLS11}.
 
 \end{itemize}

\section*{Acknowledgement} 
We thank Gregory Gutin for pointing us to QAP, and the anonymous reviewers for helpful comments.
\bibliographystyle{plainnat}
\bibliography{hom,book_pc,subiso}

\end{document}